\documentclass[12pt]{amsart}
\usepackage{amssymb}
\usepackage{graphicx}

\newtheorem{thm}{Theorem}[section]
\newtheorem{cor}[thm]{Corollary}
\newtheorem{lem}[thm]{Lemma}
\newtheorem{prop}[thm]{Proposition}
\theoremstyle{definition}
\newtheorem{defn}[thm]{Definition}
\theoremstyle{remark}

\numberwithin{equation}{section}
\newcommand{\norm}[1]{\left\Vert#1\right\Vert}
\newcommand{\abs}[1]{\left\vert#1\right\vert}
\newcommand{\set}[1]{\left\{#1\right\}}
\newcommand{\Comp}{\mathbb C}
\newcommand{\Real}{\mathbb R}
\newcommand{\Nat}{\mathbb N}
\newcommand{\Prob}{\mathbb P}
\newcommand{\Exp}{\mathbb E}

\newcommand{\eps}{\varepsilon}
\newcommand{\To}{\rightarrow}

\newcommand{\al}{\alpha}
\newcommand{\la}{\lambda}

\newcommand{\cB}{\mathcal{B}}

\newcommand{\cD}{\mathcal{D}}

\newcommand{\cF}{\mathcal{F}}
\newcommand{\cI}{\mathcal{I}}

\newcommand{\cM}{\mathcal{M}}

\newcommand{\of}[1]{\left ( #1 \right ) }

\newcommand{\zd}{\mathbb{Z}^{\mathrm{d}}}

\newcommand{\im}{\mathrm{Im}}
\newcommand{\re}{\mathrm{Re}}

\newcommand{\scalp}[2]{\langle#1|#2\rangle}

\newcommand{\supp}{\rm{supp}}

\newcommand{\txi}{\widetilde{\xi}}

\newcommand{\Leb}{\mathcal{L}}

\begin{document}

\title[]{Poisson statistics of eigenvalues in the hierarchical Anderson model.}%
\author{Evgenij Kritchevski}%
\address{Department of Mathematics and Statistics.
McGill University, 805 Sherbrooke Street West Montreal, QC, H3A
2K6 Canada}
\email{ekritc@math.mcgill.ca}



\date{
October 5, 2008}%
\begin{abstract} 
We study the eigenvalue statistics for the hieracharchial Anderson model of Molchanov  \cite{K1,K2,K3,M2,M3}.  We prove Poisson fluctuations at arbitrary disorder, when the the model has spectral dimension $\mathrm{d}< 1$. The proof is based on Minami's technique \cite{Mi} and we give an elementary exposition of the probabilistic arguments.
\end{abstract}

\maketitle
\section{Introduction}

The models discussed in this paper fall into the following general framework.
We are given a countable set $\mathbb{X}$, a bounded self-adjoint operator $H_0$ acting on the Hilbert space $l^2(\mathbb{X})$ and  a random potential $V_\omega$ acting diagonally on $l^2(\mathbb{X})$:
$$(V_\omega \psi)(x)=\omega(x)\psi(x), \qquad\psi\in l^2(\mathbb{X}), x\in \mathbb{X}.$$
Here $(\omega(x))_{x\in \mathbb{X}}$ denote independent identically distributed (i.i.d.) random variables with a bounded density $\gamma$. Hence the random parameter $\omega$ is an element of the probability space $(\Omega,\cF,\Prob)$, where $\Omega=\Real^{\mathbb{X}}$, $\cF$ is the product Borel $\sigma$-algebra on $\Omega$ and the probability measure is $\Prob=\times_{x\in\mathbb{X}}\gamma(t)dt$.  We consider the random discrete Schr\"{o}dinger operator
$$H_\omega=H_0+V_\omega.$$
The finite volume approximations to $H_\omega$ are given by an increasing sequence $(B_k)_{k\geq 1}$ of finite subsets of $\mathbb{X}$, $\bigcup_{k\geq 1}B_k=\mathbb{X}$, and a corresponding sequence of operators $(H_k^{\omega})_{k\geq 1}$ approximating $H_\omega$, such that the subspace $l^2(B_k)$ is invariant for $H_k^{\omega}$. We are interested in the asymptotic behavior of the random eigenvalues
$$e^{\omega,k}_1\leq e^{\omega,k}_2\leq\cdots\leq e^{\omega,k}_{\abs{B_k}},$$
of $H_k^{\omega}\upharpoonright l^2(B_k)$ as $k\To\infty$. Usually, the first step is to prove that there is a nonrandom probability measure $\mu^{av}$ on $\Real$ such that, with probability one, the random normalized eigenvalue counting measure
\begin{equation}\label{defn_eigenvaluecounting}
\mu_k^{\omega}=\abs{B_k}^{-1}\sum_{j=1}^{\abs{B_k}}\delta(e^{\omega,k}_j),
\end{equation}
converges $\mu^{av}$ in the weak-* topology as $k\To\infty$. The measure $\mu^{av}$ is called the density of states for $H_\omega$. For large $k$, the number of eigenvalues in a small interval $(e-\eps,e+\eps)$ around a point $e\in\Real$ is then typically of the order of $\abs{B_k}\mu^{av}((e+\eps,e-\eps))$. The fine eigenvalue statistics near $e$ are then captured by the rescaled point measure 
\begin{equation}\label{defn_rescaledmeasure}
\xi_k^{\omega,e}=\sum_{i=1}^{\abs{B_k}}\delta(\abs{B_k}(e^{\omega,k}_i-e)).
\end{equation}
Minami's technique \cite{Mi} is a method allowing to prove that, in appropriate situations, $\xi_k^{\omega,e}$ is asymptotically a Poisson point process as $k\To\infty$. This means that for disjoint Borel sets $A_1,A_2,\cdots,A_m\subset\Real$, the corresponding numbers of rescaled eigenvalues in each of the sets, $$\xi_k^{\omega,e}(A_1), \xi_k^{\omega,e}(A_2),\cdots, \xi_k^{\omega,e}(A_m),$$ are approximately independent Poisson random variables and hence the eigenvalues near $e$ are uncorrelated.

Minami originally considered the Anderson tight-binding model on $\zd$. In this case $\mathbb{X}=\zd$ and $H_0$ is the discrete Laplacian: 
\begin{equation}\label{def_discretelaplacian}
(H_0\psi)(x)=\sum_{\abs{y-x}=1}\psi(y),\qquad\psi\in l^2(\zd),x\in\zd,
\end{equation}
where $\abs{x-y}=\sum_{j=1}^{\mathrm{d}}\abs{x_j-y_j}$. He proved Poisson statistics of eigenvalues in the localized regime (\cite{Mi,KN}). Minami's method has its origins in Molchanov's paper \cite{M1}, where the first rigorous proof of the absence of energy level repulsion is given for a continuous one-dimensional model. After Minami's paper \cite{Mi}, the technique and its variations have been used to prove Poisson statistics of eigenvalues for different models \cite{AW,BHS,KN,KS,S}. In this paper, we combine existing and new results to prove Poisson statistics of eigenvalues for the hierarchical Anderson model (the precise definition of the model and the statement of our results are given in section 3).
 
The probabilistic part of Minami's technique shared by most models is based on the theory of infinitely divisible point processes. As a result, one sometimes has to go though a substantial body of material also concerned with other questions e.g. \cite{Ka, DV} in order to extract the necessary results. One of our goals is to give a self-contained  elementary exposition of the probabilistic part, only assuming standard material taught in a first graduate course on probability. The spectral part of the technique is based on decoupling, i.e. on approximating $H_\omega^k$ by a direct sum of a large number of statistically independent infinitesimal components. The analysis is specific to each model and the decoupling is possible only in an appropriate regime.

In section 2, we discuss the necessary probabilistic preliminaries on Poisson point processes. In section 3, we introduce the hierarchical Anderson model and we provide a complete proof of Poisson statistics of eigenvalues in the regime where the model has spectral dimension $\mathrm{d}<1$. In the appendix, we outline, within our framework, Minami's original proof of Poisson statistics of eigenvalues for the Anderson model on $\zd$ in the localized regime.

\textbf{Acknowledgements.} We are grateful to Vojkan Jaksic for suggesting this research project. We benefited from discussions with Michael Aizenman, Vojkan Jaksic, Rowan Killip, Stas Molchanov and Mihai Stoiciu. This work was supported by FQRNT, ISM and McGill Majors grants. 
\section{Probabilistic Preliminaries}

\subsection{Why the Poisson distribution.} 
The Poisson distribution with parameter $\la$ is the discrete probability measure $\Prob_\la$ on $\Nat=\set{0,1,2,\cdots}$ given by
$$\Prob_\la=e^{-\la}\sum_{r\in\Nat}\frac{\la^r}{r!}\delta(r).$$
The simplest example where the Poisson distribution appears naturally in connection with the rescaled measure $\xi_k^{\omega,e}$ is the trivial case of a random discrete Schr\"{o}dinger operator: $\mathbb{X}=\set{1,2,\cdots}$,$H_0=0$ and the finite volume approximations are $B_k=\set{1,\cdots,k}$, $H_k^{\omega}=H_\omega\upharpoonright l^2(B_k)$. Then $H_k^{\omega}\upharpoonright l^2(B_k)$ has statistically independent eigenvalues $\set{\omega(x)}_{x\in B_k}$ and it follows from Kolmogorov's strong law of large numbers that for every Borel set $A\subset\Real$, 
$$\lim_{k\To\infty}\mu_k^{\omega}(A)=\mu^{av}(A)=\int_A \gamma(t)dt,$$
for $\Prob$-a.e. $\omega\in\Omega$. 

We denote by $\Leb$ the Lebesgue measure on $\Real$. Let us assume that $\gamma$ is continuous at a point $e\in\Real$ and that $\gamma(e)>0$. If $A_1,A_2,\cdots, A_m\subset\Real$ are disjoint bounded Borel sets, then the random vector
$$[\xi_k^{\omega,e}(A_1), \xi_k^{\omega,e}(A_2),\cdots, \xi_k^{\omega,e}(A_m)],$$
has a multinomial distribution

$$\Prob\set{\xi_k^{\omega,e}(A_1)=r_1,\xi_k^{\omega,e}(A_2)=r_2,\cdots,\xi_k^{\omega,e}(A_m)=r_m}$$
$$=\frac{k!}{r_1!r_2!\cdots r_{m+1}!}q_{k,1}^{r_1}q_{k,2}^{r_2}\cdots q_{k,m+1}^{r_{m+1}},\qquad r_s=0,\cdots,k,\sum_{s=1}^{m+1}r_s=k,$$
where
$$q_{k,s}=\Prob\set{k(\omega(1)-e))\in A_s} =\int_{e+k^{-1}A_s}\gamma(t)dt,\qquad s=1,\cdots,m+1,$$
and $A_{m+1}=\Real\backslash(\bigcup_{s=1}^{m}A_s)$. Continuity of $\gamma$ at $e$ yields that $$\lim_{k\To\infty}k q_{k,s}=\gamma(e)\Leb(A_s),$$ and hence
$$\lim_{k\To\infty}\Prob\set{\xi_k^{\omega,e}(A_1)=r_1,\xi_k^{\omega,e}(A_2)=r_2,\cdots,\xi_k^{\omega,e}(A_m)=r_m}=\prod_{s=1}^{m}\Prob_{\la_s}(\set{r_s}),$$
with $\la_s=\gamma(e)\Leb(A_s)$.  Hence the random variables $\xi_k^{\omega,e}(A_s),s=1,\cdots, m$ are asymptotically independent and have Poisson distributions $\Prob_{\la_s}$. 

In nontrivial situations, the operator $H_0\neq 0$ introduces statistical dependence to eigenvalues of $H_k^{\omega}\upharpoonright l^2(B_k)$ and therefore the analysis of the rescaled measure $\xi_k^{\omega,e}$ is more involved. However, if the dependence introduced by $H_0$ is not too big in a suitable sense, then Minami's method allows to show that $\xi_k^{\omega,e}(A_s),s=1,\cdots, m$ are still asymptotically independent Poisson random variables. In the next subsection, we discuss a general limit theorem needed for Minami's method.
\subsection{Poisson point process and Grigelionis' limit theorem}
Although $\xi_k^{\omega,e}$ as well as the other measures of interest to us are on $\Real$, we discuss, for sake of clarity, the general situation of random point measures on a metric space $S$. We equip $S$ with the Borel $\sigma$-algebra $\cB_S$, i.e. the
$\sigma$-algebra generated by open sets. We denote by $\cM$ the set of all nonnegative Borel measures $\mu$ on $(S,\cB_S)$ such that $\mu(A)<\infty$ for every bounded Borel set $A\subset S$. A measure $\mu\in\cM$ is called a point measure if $\mu$ can be written in the form
$$\mu=\sum_{j\in J} \delta(x_j), \qquad x_j\in S,$$
where $J$ is a countable index set. If $\mu\in\cM$, then we must have $\mu(A)\in\mathbb{N}$ for every bounded Borel $A\subset S$. We denote by $\cM_p$ the set of all
point measures on $(S,\cB_S)$. A point process on $S$ is map $\omega\To \mu^{\omega}$ from some probability space $(\Omega,\cF,\Prob)$ to $\cM_p$ such that for every  bounded Borel $A\subset S$, the map $\omega\To \mu^{\omega}(A)$ is measurable. If $\mu^\omega$ is a point process, then the map
$$\nu(B)=\Exp\mu^\omega(B), \qquad B\in\cB_S,$$
defines a measure on $(S,\cB_S)$. The measure $\nu$ is called the intensity measure of the process $\mu^\omega$.
\begin{defn}\label{defnPoisson} Let $\nu\in\cM$. A Poisson point process on $S$ with intensity $\nu$ is a point process $\xi^\omega$ with the following properties:
\begin{enumerate}
\item for every bounded Borel set $A\subset S$, the random variable $\xi^\omega(A)$ has a Poisson distribution with parameter $\nu(A)$.
\item given disjoint bounded Borel sets $A_1,A_2,\cdots, A_m$ in $S$, the
random variables $\xi^\omega(A_1),\xi^\omega(A_2),\cdots,\xi^\omega(A_m)$ are independent.
\end{enumerate}
\end{defn}

It can be shown \cite{Ki} that given any $\nu\in \cM$, there exists a Poisson process on $S$  with intensity $\nu$, constructed on a suitable probability space.
Poisson point process is an idealized model of noninteraction and the point process $\xi_k^{\omega,e}$ in the study of eigenvalue statistics never exactly verifies conditions (1) and (2) of definition \ref{defnPoisson}.
\begin{defn} A sequence $\xi_k^{\omega}$ of point processes on $S$, defined on the same probability space, is said to converge to a Poisson point process on $S$ with intensity $\nu\in\cM$ if for any given disjoint bounded Borel sets $A_1,A_2,\cdots,A_m$ in $S$, we have
\begin{equation}\label{DefConvergenceToPoisson}
\lim_{k\To\infty}\Prob\set{\xi_k^\omega(A_1)=r_1,\xi_k^\omega(A_2)=r_2,\cdots,\xi_k^\omega(A_m)=r_m}=\prod_{s=1}^{m}\Prob_{\nu(A_s)}(\set{r_s}),
\end{equation}
for all $r_1,r_2,\cdots,r_m\in\Nat$.
\end{defn}
Hence, in the previous subsection, the sequence of point processes $\xi_k^{\omega,e}$ on $\Real$ converges to a Poisson process on $\Real$ with intensity $\gamma(e)\Leb$. In general, it can be difficult to verify the condition \eqref{DefConvergenceToPoisson} directly and it is more convenient to verify an equivalent condition in terms of the characteristic functions, namely
\begin{equation}\label{DefConvergenceToPoissonFourier}
\lim_{k\To\infty}\Exp e^{i\sum_{s=1}^{m} t_s\xi_k^\omega(A_s)}=\prod_{s=1}^{m}\exp\of{\nu(A_s)(e^{it_s}-1)},
\end{equation}
for all $t_1,t_2,\cdots,t_m\in\Real$. Both \eqref{DefConvergenceToPoisson} and \eqref{DefConvergenceToPoissonFourier} are equivalent to the usual definition of convergence in law for random vectors in $\Nat^m$.

The basic limit theorem guaranteeing the convergence of a sequence of point processes to a Poisson point processes is due to Griegelionis \cite{G}. Originally formulated for step processes on $\Real$, Grigelionis' theorem remains valid in more general settings and in our case it translates to:
\begin{thm}\label{Grigelionis} (Grigelionis, 1963) Let $(n_k)_{k\geq 1}$ be a natural subsequence, let for each $k\geq 1$, 
$\xi_{k,1}^{\omega},\xi_{k,2}^{\omega},\cdots,\xi_{k,n_k}^{\omega}$ be independent point processes on $S$ and let  
$$\xi_k^{\omega}=\sum_{j=1}^{n_k}\xi_{k,j}^{\omega}.$$
Let $\nu\in\cM$ and assume that  for every bounded Borel set $A\subset S$, we have

$$(1)\qquad \lim_{k\To\infty}\max _{1\leq j\leq n_k}\Prob\set{\xi_{k,j}^{\omega}(A)\geq 1}=0,$$

$$(2)\qquad \lim_{k\To\infty}\sum_{j=1}^{n_k}\Prob\set{\xi_{k,j}^{\omega}(A)\geq 1}=\nu(A),$$
and
$$(3)\qquad \lim_{k\To\infty}\sum_{j=1}^{n_k}\Prob\set{\xi_{k,j}^{\omega}(A)\geq 2}=0.$$
Then $\xi_k^{\omega}$ converges to a Poisson point process on $S$ with intensity $\nu$.
\end{thm}
Theorem \ref{Grigelionis} is well-known and can be found in the literature e.g. \cite{DV, Ka} as a corollary of more general results on point processes. For completeness, we include a self-contained proof here, following the original arguments of \cite{G}.
\begin{proof}
We use the standard notation $ab=\sum_{s=1}^{m}a_sb_s$, for $a,b\in \Real^m$ and $\abs{\al}=\sum_{s=1}^{m}\al_s$ for $\al\in\mathbb{N}^m$. We denote by $\set{e_s}_{s=1}^{m}$ the standard basis vectors of $\Real^m$. Let $A_1,A_2,\cdots,A_m$ be given disjoint bounded Borel sets in $S$. Let $X_k^\omega$ be the random vector
$$X_k^\omega=[\xi_k^\omega(A_1),\xi_k^\omega(A_2),\cdots,\xi_k^\omega(A_m)],$$
and let $\phi_k:\Real^m\To\Comp$ be the corresponding characteristic function
$$\phi_k(t)=\Exp e^{itX_k^{\omega}},\qquad t\in \Real^m.$$
According to \eqref{DefConvergenceToPoissonFourier}, we have to show that for all $t\in\Real^m$,
\begin{equation}\label{eqproofGrig}
\lim_{k\To\infty}\phi_{k}(t)=\prod_{s=1}^{m}\exp\of{\nu(A_s)(e^{it_s}-1)}.
\end{equation}
We set
$$X_{k,j}^\omega=[\xi_{k,j}^\omega(A_1),\xi_{k,j}^\omega(A_2),\cdots,\xi_{k,j}^\omega(A_m)],$$
$$\phi_{k,j}(t)=\Exp e^{itX_{k,j}^{\omega}},\qquad t\in \Real^m,$$
and $$A=\bigcup_{s=1}^m A_s.$$ By assumption (1), there is a $k_0$ such that for $k\geq k_0$, $$\max _{1\leq j\leq n_k}\Prob\set{\xi_{k,j}^{\omega}(A)\geq 1}<1/4.$$ Hence for $k\geq k_0$ and $1\leq j\leq n_k$,
$$\abs{\sum_{\abs{\al}\geq 1}\Prob\set{ X_{k,j}^\omega=\al}(e^{i\al
t}-1)}\leq 2 \sum_{\abs{\al}\geq 1}\Prob\set{ X_{k,j}^\omega=\al}=2\Prob\set{\xi_{k,j}^\omega(A)\geq 1}<1/2,$$
and we can write

\begin{equation}
\begin{split}
\phi_{k,j}(t)&=1+\sum_{\abs{\al}\geq 1}\Prob\set{X_{k,j}^{\omega}=\al}(e^{i\al t}-1)
\\ &=\exp\of{\sum_{\abs{\al}\geq 1}\Prob\set{X_{k,j}^{\omega}=\al}(e^{i\al t}-1)+E_{k,j}},
\end{split}
\end{equation}
where $$E_{k,j}=f\of{\sum_{\abs{\al}\geq 1}\Prob\set{X_{k,j}^{\omega}=\al}(e^{i\al
t}-1)},$$
and $f(z)=\log(1+z)-z$. The function $f$ is analytic in the open disk $\set{\abs{z}<1}$ and 
\begin{equation}\label{boundforlog}
\abs{f(z)}\leq C\abs{z}^2\qquad \textrm{ for } \abs{z}<1/2,
\end{equation}
where $0<C<\infty$ is a numerical constant. Next, we write
\begin{equation}
\begin{split}
\sum_{\abs{\al}\geq 1}\Prob\set{X_{k,j}^{\omega}=\al}(e^{i\al t}-1)&=\sum_{\abs{\al}=1}\Prob\set{X_{k,j}^{\omega}=\al}(e^{i\al t}-1) + F_{k,j}
\\
&=\sum_{s=1}^{m}\Prob\set{X_{k,j}^{\omega}=e_s}(e^{i t_s}-1) + F_{k,j}
\\
&=\sum_{s=1}^{m}\Prob\set{\xi_{k,j}^{\omega}(A_s)=1}(e^{i t_s}-1) + G_{k,j}+F_{k,j}
,
\end{split}
\end{equation}
where
$$F_{k,j}=\sum_{\abs{\al}\geq
2}\Prob\set{X_{k,j}^{\omega}=\al}(e^{i\al t}-1),$$
and
$$G_{k,j}=\sum_{s=1}^{m}\of{\Prob\set{X_{k,j}^{\omega}=e_s}-\Prob\set{\xi_{k,j}^{\omega}(A_s)=1}}(e^{it_s}-1).$$
Hence, 
$$\phi_{k,j}(t)=\exp\of{\sum_{s=1}^{m}\Prob\set{\xi_{k,j}^{\omega}(A_s)=1}(e^{i t_s}-1)+H_{k,j}},$$
where
$$H_{k,j}=E_{k,j}+F_{k,j}+G_{k,j}.$$
We then have, by independence, that
\begin{equation}\label{eq_poiss_factor}
\begin{split}
\phi_k(t)&=\prod_{j=1}^{n_k}\phi_{k,j}(t)
\\
&=\exp\of{\sum_{s=1}^{m}\of{\sum_{j=1}^{n_k}\Prob\set{\xi_{k,j}^{\omega}(A_s)=1}}(e^{i
t_s}-1)+\sum_{j=1}^{n_k}H_{k,j}}
\end{split}
\end{equation}
The assumptions (2) and (3) imply that
\begin{equation}\label{eqstar}
 \lim_{k\To\infty}\sum_{j=1}^{n_k}\Prob\set{\xi_{k,j}^{\omega}(A_s)=1}=\nu(A_s).
 \end{equation}
We claim that
\begin{equation}\label{eqstarstar}
\lim_{k\To\infty}\sum_{j=1}^{n_k}H_{k,j}=0.
\end{equation}
If \eqref{eqstarstar} holds, then \eqref{eqstar}, \eqref{eqstarstar} and \eqref{eq_poiss_factor} together yield the desired conclusion \eqref{eqproofGrig}
and we are done. We now prove \eqref{eqstarstar}. We have
\begin{equation}\label{eqboundF}
\abs{F_{k,j}}\leq 2\Prob\set{\xi_{k,j}^{\omega}(A)\geq 2},
\end{equation}
and the bound \eqref{boundforlog} yields
\begin{equation}\label{eqboundE}
\abs{E_{k,j}}\leq C\of{2\sum_{\abs{\al}\geq 1}\Prob\set{X_{k,j}^{\omega}=\al}}^2= 4C \of{\Prob\set{\xi_{k,j}^{\omega}(A)\geq 1}}^2.
\end{equation}
To estimate $\abs{G_{k,j}}$, note that
$$\set{X_{k,j}^{\omega}=e_s}\subset \set{\xi_{k,j}^{\omega}(A_s)=1},$$
and
$$\of{ \set{\xi_{k,j}^{\omega}(A_s)=1}\backslash \set{X_{k,j}^{\omega}=e_s}}\subset \set{\xi_{k,j}^{\omega}(A)\geq 2}.$$
Hence
\begin{equation}\label{eqboundG}
\abs{G_{k,j}}\leq 2m \Prob\set{\xi_{k,j}^{\omega}(A)\geq 2}.
\end{equation}
We now combine the bounds \eqref{eqboundE}, \eqref{eqboundF} and \eqref{eqboundG} to get
\begin{equation*}
\begin{split}
\abs{\sum_{j=1}^{n_k}H_{k,j}}\leq & (2m+2)\sum_{j=1}^{n_k}\Prob\set{\xi_{k,j}^{\omega}(A)\geq 2}\\
&+4C\of{\max_{1\leq j\leq n_k}\Prob\set{\xi_{k,j}^{\omega}(A)\geq 1}}\sum_{j=1}^{n_k}\Prob\set{\xi_{k,j}^{\omega}(A)\geq 1}.
\end{split}
\end{equation*}
The assumptions (1),(2) and (3) imply that the right hand side of last inequality
converges to zero as $k\To\infty$, completing the proof.
\end{proof}

\subsection{Corollaries of Grigelionis' limit theorem}
For the point processes $\xi^\omega$ on $S=\Real$ arising in the study of eigenvalue statistics, it is sometimes more natural to obtain information about the Poisson integrals $\int_\Real \im(t-z)^{-1}d\xi^\omega(t)$, $\im z>0$, rather than about the events $\set{\xi^\omega(A)\geq 1}$ and $\set{\xi^\omega(A)\geq 2}$. In this subsection, we replace the conditions (2) and (3) of Theorem \ref{Grigelionis} by sufficient conditions in terms of the Poisson integrals. We refer the reader to \cite{J} for the general theory of Poisson integrals and their applications to spectral theory.

For a positive Borel measure $\mu$ on $S$ and a Borel function
$f:S\To [0,\infty)$, we set $$\cI(\mu,f)=\int_{t\neq
t'}f(t)f(t')d\mu(t) d\mu(t').$$ If $\mu=\sum_{j}\delta(t_j)$ is a
point measure on $S$ and $f(t)=1_A(t)$ is the indicator function of a bounded Borel set $A\subset S$, then we have
$$\cI(\mu,1_A)=\sum_{i\neq j}1_A(t_i)1_A(t_j)=\mu(A)(\mu(A)-1),$$
and therefore $\cI(\mu,1_A)\neq 0\Leftrightarrow \mu(A)\geq 2.$ If
$\xi^\omega$ is a point process on $S$, then 
\begin{equation*}
\begin{split}
\sum_{l\geq 2}\Prob\set{\xi^\omega(A)\geq l}&=\sum_{l\geq 2}(l-1)\Prob\set{\xi^\omega(A)=l}
\\
&\leq \sum_{l\geq
2}l(l-1)\Prob\set{\xi^\omega(A)=l}
\\
&=\Exp \cI(\xi^\omega,1_A).
\end{split}
\end{equation*}
Since
$$\Prob\set{\xi^\omega(A)\geq 1}=\Exp\xi^\omega(A)-\sum_{l\geq 2}\Prob\set{\xi^\omega(A)\geq l},$$
we conclude that the conditions
$$(2')\qquad \lim_{k\To\infty}\sum_{j=1}^{n_k}\Exp\xi_{k,j}^{\omega}(A)=\nu(A),$$
and
$$(3')\qquad \lim_{k\To\infty}\sum_{j=1}^{n_k}\Exp \cI(\xi_{k,j}^\omega,1_A)=0,$$
together imply conditions (2) and (3) of Theorem \ref{Grigelionis}. The next step is to replace, in (2') and (3'), the quantity $\Exp\xi_{k,j}^{\omega}(A)$ by $\Exp\int f d\xi_{k,j}^{\omega}$ for $f$ in a sufficiently rich class of functions $F$.
\begin{thm}
For each $k\geq 1$, let $\xi_{k,1}^{\omega}, \xi_{k,2}^{\omega},\cdots,\xi_{k,n_k}^{\omega}$ be point processes on $S$ and let $\xi_k^{av}=\sum_{j=1}^{n_k}\Exp\xi_{k,j}^{\omega}$. Let $\nu\in\cM$. Suppose that there is a measure $\mu\in\cM$ s.t. that $\nu$ and $(\xi_k^{av})_{k\geq 1}$ are absolutely continuous with respect to $\mu$, with uniformly bounded densities, i.e. there is a constant $0<C<\infty$ such that for all bounded Borel sets $A\subset S$,
 $$\nu(A)\leq C \mu(A),$$ and  
 $$\xi_k^{av}(A) \leq C \mu(A),\qquad k\geq 1.$$
Suppose that $F\subset L_1(S,\mu)$ is a family of functions such
that finite linear combinations of functions in $F$ are dense in
$L_1(S,\mu)$ and such that for every bounded Borel set $A\subset S$, there exists $f\in F$ with $f\geq 1_A$. 
Suppose that for all $f\in\cF$, we have 
$$(2'')\qquad \lim_{k\To\infty}\int f d\xi_k^{av}=\int f d\nu,$$
and
$$(3'')\qquad \lim_{k\To\infty}\sum_{j=1}^{n_k}\Exp \cI(\xi_{k,j}^{\omega},f)=0.$$
Then (2') and (3') hold for all bounded Borel sets $A\subset S$.
\end{thm}
\begin{proof} Let $A$ be a bounded Borel set. Let $\eps>0$. There is a finite linear combination $g=\sum_ic_i f_i$, $f_i\in F$, with $\int\abs{g-1_A}d\mu<\eps$. Then $\abs{\int g d\nu-\nu(A)}<C\eps$ and $\abs{\int g d\xi_k^{av}-\xi_k^{av}(A)}<C\eps$. Since $\lim_{k\To\infty}\int g d\xi_k^{av}=\int gd\nu$, we have
$$\nu(A)-2C\eps\leq \liminf_{k\To\infty}\xi_k^{av}(A)\leq\limsup_{k\To\infty}\xi_k^{av}(A)\leq \nu(A)+2C\eps,$$
and (2') is obtained after letting $\eps\downarrow 0$. Now let $f\in F$ be such that $f\geq 1_A$. Since, $\cI(\xi_{k,j}^\omega, 1_A)\leq I(\xi_{k,j}^\omega,f)$, (3') follows from (3'').
\end{proof}

The special case when $S=\Real$, $\mu=\Leb$ is the Lebesgue measure on $\Real$, $\nu=\la\Leb$ for a $\la>0$ and $F$ is the family of functions
 $\set{\im (t-z)^{-1}}_{\im z>0}$ yields 
\begin{thm}\label{GrigelionisPoisson} Let $(n_k)_{k\geq 1}$ be a natural subsequence, let for each $k\geq 1$, 
$\xi_{k,1}^{\omega},\xi_{k,2}^{\omega},\cdots,\xi_{k,n_k}^{\omega}$ be independent point processes on $\Real$ and let  
$$\xi_k^{\omega}=\sum_{j=1}^{n_k}\xi_{k,j}^{\omega}.$$
We make the following four hypotheses:
\newline\indent (H0): there is a constant $0<C<\infty$ such that for all $k\geq 1$ and every bounded Borel set $A\subset \Real$,
$$\sum_{j=1}^{n_k}\Exp\xi_{k,j}^{\omega}(A)\leq C \Leb(A).$$
\newline\indent (H1): for every bounded Borel set $A\subset\Real$,
 $$\lim_{k\To\infty}\max _{1\leq j\leq n_k}\Prob\set{\xi_{k,j}^\omega(A)\geq 1}=0.$$
\newline\indent (H2): there is a constant $0<\la<\infty$ such that
for $\im z>0$,
$$\lim_{k\To\infty}\sum_{j=1}^{n_k}\Exp \int_\Real \im (t-z)^{-1}d\xi_{k,j}^\omega(t)=\pi \la.$$
\newline\indent (H3): for $\im z>0$,
$$\lim_{k\To\infty}\sum_{j=1}^{n_k}\Exp \int_{t\neq t'} \im (t-z)^{-1}\im (t'-z)^{-1}d\xi_{k,j}^\omega(t)d\xi_{k,j}^\omega(t')=0.$$
Then $\xi_k^\omega$ converges to a Poisson point process on $\Real$ with intensity $\la \Leb$.
\end{thm}
Theorem \ref{GrigelionisPoisson} is implicitly derived in \cite{Mi} and is suitable for applications to eigenvalue statistics of general random discrete Schr\"{o}dinger operators.

\section{Poisson statistics of eigenvalues in the hierarchical Anderson model} 
\subsection{Definition of the model and its basic properties} In this subsection, we review the definition and the basic properties of the hierarchical Anderson model. For additional information, we refer the reader to \cite{K1,K2,K3,M2,M3}.
Theorems \ref{free_lap_sp} and \ref{hierarchical_sp} collect, for reference purposes, the main known results on the hierarchical Anderson model and are stated without proof. 

We consider the set $\mathbb{X}=\set{0,1,2,\dots}$. Given an integer $n\geq 2$, $\mathbb{X}$ has a metric space structure with the distance $d:\mathbb{X}\times \mathbb{X}\To[0,\infty)$
$$d(x,y)=\min\set{r: q(x,n^r)=q(y,n^r)}, $$
where $q(x,n^r)$ denotes the quotient of the division of $x$ by $n^r$. The closed ball with center $x$ and radius $r$ is denoted by 
$$B(x,r)=\set{y\in \mathbb{X}: d(x,y)\leq r}.$$ The main property of $d$ is that two closed balls of the same radius are either disjoint or identical, and that each $B(x,r+1)$ is a disjoint union of $n$ balls of radius $r$. 

For $x\in \mathbb{X}$, the unit vector $\delta_x\in l^2(\mathbb{X})$ denotes the Kronecker delta function at $x$: $\delta_x(x)=1$ and $\delta_x(y)=0$ for $y\neq x$. For each integer $r\geq 1$, we set $E_r:l^2(\mathbb{X})\To l^2(\mathbb{X})$, 
$$(E_r\psi)(x)=n^{-r}\sum_{d(y,x)\leq r}\psi(y).$$ Thus $E_r$ is the orthogonal projection onto the subspace of $l^2(\mathbb{X})$ consisting of functions that are constant on every closed ball of radius $r$.
The hierarchical Laplacian is then defined by the formula
$$\Delta=\sum_{r=1}^{\infty}p_r E_r,$$
where $(p_r)_{r\geq 1}$ is a given sequence such that $p_r>0$
and $\sum_{r=1}^\infty p_r =1$.  We assume that
$$\frac{C_1}{\rho^r}\leq p_r\leq \frac{C_2}{\rho^r},$$
for some fixed constants $\rho>1,C_1>0,C_2>0$. The number
\begin{equation}\label{dimension}
{\rm d}={\rm d}(n,\rho)=2\frac{\log n}{\log \rho},
\end{equation}
is called the spectral dimension of $\Delta$.
The following theorem \cite{K1, M3} summarizes some of the spectral features of $\Delta$.
\begin{thm}\label{free_lap_sp}
$\Delta$ is a bounded self-adjoint operator on $l^2(\mathbb{X})$ and its spectrum consists of infinitely degenerate isolated eigenvalues
$$\la_0=0,\la_1=p_1,\la_2=p_1+p_2,\la_3=p_1+p_2+p_3,\cdots$$
and of their accumulation point $\la_\infty=1$, which is not an eigenvalue. For each $x\in \mathbb{X}$,
$$\sum_{y\in \mathbb{X}}\scalp{\delta_x}{\Delta\delta_y}=1,$$
and hence $\Delta$ generates a random walk on $\mathbb{X}$.The random walk is
recurrent when $\rm{d}\leq 2$ and transient when $\rm{d}>2$. 
\end{thm}


The hierarchical Anderson model is the random discrete Schr\"{o}dinger operator
$$H_\omega=\Delta+V_\omega,$$
as in the framework of the introduction, with $H_0=\Delta$.
If the set $\set{\omega(x):x\in \mathbb{X}}$ is unbounded, then $V_\omega$ and $H_\omega$ are unbounded self-adjoint operators with the domain 
$$\cD_\omega=\set{\psi: \sum_{x\in \mathbb{X}}\abs{\psi(x)}^2(1+\abs{\omega(x)}^2)<\infty}.$$



\begin{thm}\label{hierarchical_sp} $H_\omega$ has the following generic spectral properties.
\newline\noindent (1) \cite{K2}
 If the support of $\gamma$ is connected, $\supp(\gamma)=[a,b]$, $-\infty\leq a<b\leq\infty$, then for $\Prob$-a.e. $\omega\in\Omega$, the spectrum of $H_\omega$ is given by
$$\Sigma=\bigcup_{r=0}^{\infty}[\la_r+a,\la_r+b].$$
\newline\noindent (2) \cite{K2} If the model has spectral dimension $\rm{d} <4$ then, for $\Prob$-a.e. $\omega\in\Omega$, the spectrum of $H_\omega$ is dense pure-point in $\Sigma$.
\newline\noindent (3) \cite{M3} For any spectral dimension $\mathrm{d}<\infty$, the same conclusion as in (2) holds provided the random variables $\omega(x)$ have a Cauchy distribution, i.e. the density $\gamma(t)$ is of the special form: 
\begin{equation}\label{Cauchy}
\gamma(t)=\frac{1}{\pi}\frac{v}{(u-t)^2+v^2},
\end{equation}
for some $u\in\Real,v>0$.
\end{thm}

\subsection{Density of states.}
We denote by $C_0(\Real)$ the space of continuous functions $f:\Real\To\Comp$ vanishing at infinity, i.e. $\lim_{\abs{t}\To\infty}\abs{f(t)}=0$.
If $(\nu_k)_{k\geq 1}$ and $\nu$ are Borel probability measures on $\Real$, we say that $\nu_k$ converges to $\nu$ in the weak-* topology if for every $f\in C_0(\Real)$,
$$\lim_{k\To\infty}\int f(t)d\nu_k(t)=\int f(t)d\nu(t).$$
The finite volume approximations to $H_\omega$ are defined as follows. We fix $x_0\in \mathbb{X}$ and we consider the increasing sequence of closed balls 
$$B_k=B(x_0,k)\qquad k\geq 0.$$
Each $B_k$ has then size $\abs{B_k}=n^k$. We define $H_k^{\omega}$ to be the truncated operator 
$$H_k^{\omega}=\sum_{s=1}^{k}p_sE_s + V_\omega.$$
Note that the subspace
$$l^2(B_k)=\set{\psi\in l^2(\mathbb{X}):\psi(x)=0 \textrm{ for } x\notin B_k},$$
is invariant for $H_k^{\omega}$. The normalized eigenvalue counting measure $\mu_k^{\omega}$ is then given by \eqref{defn_eigenvaluecounting}.
The averaged spectral measure for $H_\omega$ is the unique Borel probability measure $\mu^{av}$ on $\Real$ defined by
\begin{equation}\label{defn_avgspectralmeasure}
\int f(t)d\mu^{av}(t)=\Exp\scalp{\delta_{x_0}}{f(H_\omega)\delta_{x_0}}, \qquad f\in C_0(\Real).
\end{equation}
By symmetry, $\int f(t)d\mu^{av}(t)=\Exp\scalp{\delta_{x}}{f(H_\omega)\delta_{x}}$ for all $x\in \mathbb{X}$. 
The content of the following theorem is that the averaged spectral measure $\mu^{av}$ is naturally interpreted as the density of states for $H_{\omega}$.
\begin{thm}\label{density_of_states} For $\Prob$-a.e. $\omega\in\Omega$, $\mu_k^{\omega}\To\mu^{av}$ in the weak-* topology as $k\To\infty$, i.e. there is a set $\widetilde{\Omega}\in\cF$ with $\Prob(\widetilde{\Omega})=1$ such that for all $\omega\in\widetilde{\Omega}$ and $f\in C_0(\Real)$ we have
$$\lim_{k\To\infty}\int f(t)d\mu_k^{\omega}(t)=\int f(t)d\mu^{av}(t).$$
\end{thm}
We start the proof of Theorem \ref{density_of_states} with resolvent bounds.
Since $H^{\omega}_{r}=H^{\omega}_{r-1}+p_rE_r$, the resolvent identity yields
$$(H^{\omega}_{r-1}-z)^{-1}-(H^{\omega}_{r}-z)^{-1}=p_r(H^{\omega}_{r-1}-z)^{-1}E_r(H^{\omega}_{r}-z)^{-1},$$
for $z\in\Comp\backslash\Real$.
Therefore:
\begin{equation}\label{resolvent1}
\norm{(H^{\omega}_{r-1}-z)^{-1}-(H^{\omega}_{r}-z)^{-1}}\leq \abs{\im z}^{-2}p_r,\qquad z\in\Comp\backslash\Real.
\end{equation}
Iterating \eqref{resolvent1} yields for $r<k$,
\begin{equation}\label{resolvent2}
\norm{(H^{\omega}_{r}-z)^{-1}-(H^{\omega}_{k}-z)^{-1}}\leq \abs{\im z}^{-2}\sum_{s=r+1}^{k}p_s,\qquad z\in\Comp\backslash\Real.
\end{equation}
and letting $k\To\infty$, 
\begin{equation}\label{resolvent3}
\norm{(H^{\omega}_{r}-z)^{-1}-(H_{\omega}-z)^{-1}}\leq \abs{\im z}^{-2}\sum_{s=r+1}^{\infty}p_s,\qquad z\in\Comp\backslash\Real.
\end{equation}
\begin{prop}\label{prop_density}For every $z\in\Comp\backslash\Real$ there is a set $\Omega_z\in\cF$, with $\Prob(\Omega_z)=1$
 and such that for all $\omega\in\Omega_z$,
 the difference
 $$D_{k,\omega}=\int(t-z)^{-1}d\mu_k^{\omega}(t)-\int(t-z)^{-1}d\mu^{av}(t),$$
 converges to $0$ as $k\To\infty$.
 \end{prop}
 \begin{proof}
 Let $\eps>0$ be given. We take $r=r(\eps,z)$ big enough so that
 \begin{equation}\label{decouple}
\abs{\im z}^{-2}\sum_{s=r+1}^{\infty}p_s<\eps/2.
 \end{equation} Then for $r<k$,
 \begin{equation}
 \begin{split}
D_{k,\omega}&=\abs{B_k}^{-1}\sum_{x\in B_k}\scalp{\delta_x}{(H_k^{\omega}-z)^{-1}\delta_x}- \Exp\scalp{\delta_{x_0}}{(H_{\omega}-z)^{-1}\delta_{x_0}}\\
&=\set{\abs{B_k}^{-1}\sum_{x\in B_k}\scalp{\delta_x}{\of{(H_k^{\omega}-z)^{-1}-(H_r^{\omega}-z)^{-1}}\delta_x}} \\&+
\set{\abs{B_k}^{-1}\sum_{x\in B_k}\scalp{\delta_x}{(H_r^{\omega}-z)^{-1}\delta_x}- \Exp\scalp{\delta_{x_0}}{(H_{\omega}-z)^{-1}\delta_{x_0}}}\\
&= I_{k,\omega} +II_{k,\omega}
\end{split}
\end{equation}
The bounds \eqref{resolvent2} and \eqref{decouple} yield $\abs{I_{k,\omega}}<\eps/2$. We proceed with estimating $\abs{II_{k,\omega}}$. Note that
$B_k$ is a disjoint union of $n^{k-r}$ balls of radius $r$,
$$B_k=\bigcup_{j=1}^{n^{k-r}}B_{k,j},$$
and therefore
$$l^2(B_k)=\bigoplus_{j=1}^{n^{k-r}}l^2(B_{k,j}).$$
Since each subspace $l^2(B_{k,j})$ is invariant for $H_r^{\omega}$, we can write
$$\abs{B_k}^{-1}\sum_{x\in B_k}\scalp{\delta_x}{(H_r^{\omega}-z)^{-1}\delta_x}=\frac{1}{n^{k-r}} 
\sum_{j=1}^{n^{k-r}} n^{-r}\sum_{x\in B_{k,j}}\scalp{\delta_x}{(H_r^{\omega}-z)^{-1}\delta_x},$$
and recognize that the right hand side is an average of $n^{k-r}$ identically distributed random variables. Hence, Kolmogorov's strong law of large numbers yields that there is a set $\Omega_{z,\eps}\in\cF$ with $\Prob(\Omega_{z,\eps})=1$ and such that for all $\omega\in\Omega_{z,\eps}$,
\begin{equation}\label{lln}
\lim_{k\To\infty}\abs{B_k}^{-1}\sum_{x\in B_k}\scalp{\delta_x}{(H_r^{\omega}-z)^{-1}\delta_x}=n^{-r}\sum_{x\in \widetilde{B}}\Exp \scalp{\delta_x}{(H_r^{\omega}-z)^{-1}\delta_x},
\end{equation}
where $\widetilde{B}$ is some fixed ball of radius $r$. The bounds \eqref{resolvent2} and \eqref{decouple} yield 
$$\abs{\scalp{\delta_x}{(H_r^{\omega}-z)^{-1}\delta_x}- \scalp{\delta_x}{(H_{\omega}-z)^{-1}\delta_x}}<\eps/2,$$
which combined with \eqref{lln} yields
$$\limsup_{k\To\infty}\abs{II_{k,\omega}}<\eps/2.$$ Hence for $\omega\in\Omega_{z,\eps}$,
$\limsup_{k\To\infty}\abs{D_{k,\omega}}<\eps,$
and the statement follows after taking 
$\Omega_z=\bigcap_{m=1}^{\infty}\Omega_{z,1/m}.$
\end{proof}
Theorem \ref{density_of_states} is a consequence of Proposition \ref{prop_density} and a density argument. Let $G$ be a countable dense set in $\Comp\backslash\Real$. Since any function $f\in C_0(\Real)$ can be uniformly approximated by finite linear combinations of the functions $t\To (t-z)^{-1}$, with $z$ ranging through $G$, Theorem \ref{density_of_states} follows after taking 
$\widetilde{\Omega}=\bigcap_{z\in G}\Omega_z.$
\newline \textit{Remarks on Theorem \ref{density_of_states}:} There is no restriction on the spectral dimension $\rm{d}$. Also, the theorem and the above proof remain valid without the assumption that the random variables $\omega(x)$ have a density $\gamma$.

\subsection{Fine eigenvalue statistics.}
For our study of fine eigenvalue statistics, we need the following two well-known general estimates for random discrete Schr\"{o}dinger operators. For both estimates, the density $\gamma$ plays a fundamental role. 
\begin{lem}[Wegner Estimate \cite{W}]\label{general_wegner} Let $M_0$ be any self-adjoint operator on $l^2(\mathbb{X})$ and let
$$M_\omega=M_0 + V_\omega.$$
Then for every bounded Borel measurable function $h:\Real\To [0,\infty)$ and $x\in \mathbb{X}$,
\begin{equation}\label{general_wegner_estimate}
\Exp\scalp{\delta_x}{h(M_\omega)\delta_x}\leq \norm{\gamma}_\infty\int h(t)dt.
\end{equation}
Hence, if $\nu_\omega$ is the spectral measure for $\delta_x$ and $M_\omega$ and $\nu^{av}=\Exp\nu_\omega$ is the corresponding averaged measure, then $\nu^{av}$ is absolutely continuous with respect to Lebesgue measure,
$$d\nu^{av}(t)=\upsilon(t)dt,$$
and
$$\norm{\upsilon}_{\infty}\leq\norm{\gamma}_{\infty}.$$ 
\end{lem}
\begin{lem}[Minami's Estimate \cite{Mi,GV,BHS}]\label{general_minami} Let $M_0$ be any self-adjoint operator on $l^2(\mathbb{X})$ and let
$$M_\omega=M_0 + V_\omega.$$
Then for every $x,y\in \mathbb{X}$ and $\im z >0$
\begin{equation}\label{general_minami_estimate}
\Exp\det \left( \begin{array}{cc}
\scalp{\delta_x}{\im(M_\omega -z)^{-1}\delta_x}&\scalp{\delta_x}{\im(M_\omega -z)^{-1}\delta_y}  \\
\scalp{\delta_y}{\im(M_\omega -z)^{-1}\delta_x} & \scalp{\delta_y}{\im(M_\omega -z)^{-1}\delta_y} \\
\end{array} \right) \leq \pi ^2\norm{\gamma}_{\infty}^2
\end{equation}
\end{lem}

Wegner estimate yields that $\mu^{av}$ is absolutely continuous with respect to Lebesgue measure,
$$d\mu^{av}(t)=\eta(t)dt,$$
and
$$\norm{\eta}_{\infty}\leq\norm{\gamma}_{\infty}.$$ 
If $e\in \sum$ and $\eps>0$ are given, then in view of Theorem \ref{density_of_states} we expect the number of eigenvalues of $H_k^{\omega}\upharpoonright l^2(B_k)$ in the interval $(e-\eps,e+\eps)$,
$$\#\set{i:e^{\omega,k}_i\in (e-\eps,e+\eps)},$$
to have typical size of order $\abs{B_k}\mu^{av}(e-\eps,e+\eps)$ for large $k$. The precise statistical behavior of the eigenvalues $e^{\omega,k}_j$ near $e$ is captured by the rescaled measure $\xi_k^{\omega,e}$ given by \eqref{defn_rescaledmeasure}.
We make the following regularity assumption on $e$:
for $\im z >0$, 
\begin{equation}\label{FatouPoint}
\lim_{\eps\downarrow 0}\int \im (t-e-\eps z)^{-1} \eta(t)dt=\pi\eta(e).
\end{equation}
For example, if $\eta$ is continuous at $e$, then \eqref{FatouPoint} holds. However, it is in general a difficult problem to establish continuity of $\eta$ for random discrete Schr\"{o}dinger operators. In the case of Cauchy random potential \eqref{Cauchy},   $\eta$ is known to be analytic \cite{L}. If the Fourier transform of $\gamma(t)$ decays exponentially, then it is possible \cite{CFS} to prove analyticity of $\eta$ after increasing the disorder, i.e. replacing $V_\omega$ with $\sigma V_\omega$ for a sufficiently large $\sigma$. When continuity of $\eta$ is not available, one appeals to a classical theorem in harmonic analysis (see for example \cite{Ko}), due to Fatou, guaranteeing that \eqref{FatouPoint} holds for Lebesgue almost all $e\in\Real$. We now state our main result.
\begin{thm}\label{Poisson_stat} Assume that the model has spectral dimension $\rm{d}<1$. Assume that $\eta(e)>0$ and that $e$ verifies the regularity condition \eqref{FatouPoint}. Then $\xi_k^{\omega,e}$ converges to a Poisson point process on $\Real$ with intensity $\eta(e)\Leb$.
\end{thm}
\textit{Remarks on Theorem \ref{Poisson_stat}:} We refer the reader to \cite{HM} for a discussion of the set of $e$ for which $\eta(e)>0$, in the context of the Anderson model on $\zd$. Our theorem is the analogue of Minami's results for the Anderson model on $\zd$ in dimension one as well as in the localized regime in higher dimensions (see Appendix). The proof of Poisson statistics for the hierarchical Anderson is technically simpler than the corresponding proofs for the Anderson model on $\zd$, because of the low spectral dimension assumption and because of the high degree of self-similarity of the model.

The rest of the section is devoted to the proof of Theorem \ref{Poisson_stat}.
The main idea is to approximate $H_{k}^{\omega}$ with $H_{r}^{\omega}$ for $r<k$, as in the proof of Theorem \ref{density_of_states}.
This time we choose $r$ to depend on $k$, $r=r_k$, such that 
$$\lim_{k\To\infty}\frac{r_k}{k}=c,$$
where
\begin{equation}\label{decouple_dimension}
\mathrm{d}<c<1.
\end{equation}
Let $$\widetilde{e}^{\omega,k}_1\leq \widetilde{e}^{\omega,k}_2\leq\cdots\leq \widetilde{e}^{\omega,k}_{\abs{B_k}},$$
denote the eigenvalues of $H_{r_k}^{\omega}\upharpoonright l^2(B_k)$ and let
$$\txi_k^{\omega,e}=\sum_{i=1}^{\abs{B_k}}\delta(\abs{B_k}(\widetilde{e}^{\omega,k}_i-e)),$$
be the corresponding rescaled measure near $e$. Since
$B_k$ is a disjoint union of $n^{k-r_k}$ closed balls of radius $r_k$,
$$B_k=\bigcup_{j=1}^{n^{k-r_k}}B_{k,j},$$
we have the corresponding direct sum decomposition
$$H_{r_k}^{\omega}\upharpoonright l^2(B_k)=\bigoplus_{j=1}^{n^{k-r_k}} H_{r_k}^{\omega}\upharpoonright l^2(B_{k,j}).$$
Therefore the point process $\txi_k^{\omega,e}$ is the sum of $n^{k-r_k}$ independent point processes,
$$\txi_k^{\omega,e}=\sum_{j=1}^{n^{k-r_k}} \txi_{k,j}^{\omega,e},$$
where 
$$\txi_{k,j}^{\omega,e}=\sum_{l=1}^{n^{r_k}}\delta(\abs{B_k}(\widetilde{e}^{\omega,k,j}_{l}-e)),$$
and $\widetilde{e}^{\omega,k,j}_{l}$, $l=1,\cdots,n^{r_k}$ are the eigenvalues of $H_{r_k}^{\omega}\upharpoonright l^2(B_{k,j})$.

The proof of Theorem \ref{Poisson_stat} is organized as follows. We first establish that the point processes $\xi^{\omega,e}_k$ and $\txi^{\omega,e}_k$ are asymptotically close in the following sense:
\begin{prop}\label{rescaled_measures_close} For every $f\in L_1(\Real,dt)$,
\begin{equation}\label{eq_same_limit}
\lim_{k\To\infty} \Exp\abs{\int f d\txi^{\omega,e}_k- \int f d\xi^{\omega,e}_k}=0.
\end{equation}
\end{prop}
\begin{cor}\label{Cor_rescaled_measures_close}
Let $A_1,A_2,\cdots,A_m$ be given disjoint bounded Borel sets in $\Real$. Let $X_k^\omega$ and $\widetilde{X}_k^{\omega}$ be the random vectors
$$X_k^\omega=[\xi^{\omega,e}_k(A_1),\xi^{\omega,e}_k(A_2),\cdots,\xi^{\omega,e}_k(A_m)],$$
$$\widetilde{X}_k^\omega=[\txi^{\omega,e}_k(A_1),\txi^{\omega,e}_k(A_2),\cdots,\txi^{\omega,e}_k(A_m)].$$
and let $\phi_k,\widetilde{\phi}_k:\Real^m\To\Comp$ be the corresponding characteristic functions
$$\phi_k(t)=\Exp e^{itX_k^{\omega}},\widetilde{\phi}_k(t)=\Exp e^{it\widetilde{X}_k^{\omega}},\qquad t\in \Real^m.$$
Then for all $t\in\Real^m$,
$$\lim_{k\To\infty}\abs{\phi_k(t)-\widetilde{\phi}_k(t)}=0.$$
\end{cor}
Then we establish 
\begin{prop}\label{asymptoticallyPoisson} The point process $\txi^{\omega,e}_k$ converges to a Poisson point process on $\Real$ with intensity $\eta(e)\Leb$.
\end{prop}
Proposition \ref{asymptoticallyPoisson} and Corollary \ref{Cor_rescaled_measures_close} together imply Theorem  \ref{Poisson_stat}.
The Wegner estimate plays a crucial role in the proof of Propositions \ref{rescaled_measures_close} and \ref{asymptoticallyPoisson}. For every Borel set $A\subset\Real$, we have
$\xi^{\omega,e}_k(A)=\sum_{x\in B_k}\scalp{\delta_x}{f(H^\omega_k) \delta_x},$
where $f(t)=1_A(\abs{B_k}(t-e))$. Wegner estimate \eqref{general_wegner_estimate} yields that for all $x\in B_k$,
\begin{equation}\label{wegner_rescaled}
\Exp\scalp{\delta_x}{f(H^\omega_k) \delta_x}\leq \norm{\gamma}_{\infty}\int f(t)dt=\norm{\gamma}_\infty \abs{B_k}^{-1}\Leb(A).
\end{equation}
Summing \eqref{wegner_rescaled} over all $x\in B_k$ yields
\begin{equation}\label{avgxi}
\Exp\xi^{\omega,e}_k(A)\leq \norm{\gamma}_\infty \Leb(A).
\end{equation}
Similarly
\begin{equation}\label{avgtxi}
\Exp\txi^{\omega,e}_k(A)\leq \norm{\gamma}_\infty \Leb(A).
\end{equation}

\textit{Proof of Proposition \ref{rescaled_measures_close}.}
Step 1: We first prove \eqref{eq_same_limit} for the family of functions $$g_z(t)=\im (t-z)^{-1}, \qquad \im z> 0.$$ Setting
\begin{equation}\label{def_zk}
z_k=e+\abs{B_k}^{-1}z,
\end{equation} we have
$$\int g_z d\txi^{\omega,e}_k- \int g_z d\xi^{\omega,e}_k=\abs{B_k}^{-1}\im
\sum_{x\in B_k}\scalp{\delta_x}{\of{(H_{r_k}^{\omega}-z_k)^{-1}-(H_k^{\omega}-z_k)^{-1}}\delta_x}.$$
Hence
$$\abs{\int g_z d\txi^{\omega,e}_k- \int g_z d\xi^{\omega,e}_k}\leq\abs{\im z_k}^{-2}\sum_{s=r_k+1}^{\infty}p_s=O\of{\frac{n^{2k}}{\rho^{c k}}}.$$
Now \eqref{decouple_dimension} and \eqref{dimension} imply $\frac{n^2}{\rho^c}<1$ and \eqref{eq_same_limit} follows.
Step 2: To prove \eqref{eq_same_limit} for general $f\in L_1(\Real,dt)$, note that $span\set{g_z, \im z>0}$ is dense in $L^1(\Real,dt)$. Hence given $\eps>0$, there is a finite linear combination $$g(t)=\sum_{j=1}^{p}a_j \im (t-z^{(j)})^{-1},\qquad\im z^{(j)}>0,$$ with 
$$\int_\Real\abs{f(t)-g(t)}dt\leq \eps.$$ The triangle inequality 
\begin{equation*}
\begin{split}
\Exp\abs{\int fd\xi^{\omega,e}_k - \int fd\txi^{\omega,e}_k(t)} &\leq \Exp \int\abs{f-g}fd\xi^{\omega,e}_k \\
&+ \Exp \abs{\int gd\xi_n-\int gd\txi^{\omega,e}_k}+\Exp\int\abs{g-f}fd\txi^{\omega,e}_k,
\end{split}
\end{equation*}
together with Step 1 and the bounds \eqref{avgxi} and \eqref{avgtxi} imply
$$\limsup_{k\To\infty}\Exp\abs{\int fd\xi^{\omega,e}_k - \int fd\txi^{\omega,e}_k(t)}\leq 2\norm{\gamma}_\infty\eps,$$ and \eqref{eq_same_limit}  follows after letting $\eps\downarrow 0$.
$\Box$
\newline\textit{Proof of Proposition \ref{asymptoticallyPoisson}.}
I suffices to show that $\txi^{\omega,e}_k$ and the $\txi^{\omega,e}_{k,j}$ verify the four hypotheses of Theorem \ref{GrigelionisPoisson}. 
\newline (H0) holds because of the bound \eqref{avgtxi}. 
\newline (H1): we need to to establish that for every bounded Borel set $A\subset\Real$, 
\begin{equation}\label{verif_infinitesimal}
\lim_{k\To\infty}\max_{1\leq j\leq n^{k-r_k}}\Prob(\txi_{k,j}^{\omega,e}(A)\geq 1)=0.
\end{equation}
\begin{proof} Chebyshev's inequality and the bound \eqref{wegner_rescaled} yield 
\begin{equation*}
\begin{split}
\Prob(\txi_{k,j}^{\omega,e}(A)\geq 1) &\leq \Exp\txi_{k,j}^{\omega,e}(A) \\
 &\leq \frac{\abs{B_{k,j}}}{\abs{B_k}}\norm{\gamma}_\infty \Leb(A) \\
 &= n^{r_k-k}\norm{\gamma}_\infty \Leb(A),
\end{split}
\end{equation*}
and \eqref{verif_infinitesimal} follows.
\end{proof}
\noindent (H2): We need to establish that for all $\im z>0$,
$$\lim_{k\To\infty}\Exp\int \im(t-z)^{-1}d\txi^{\omega,e}_k(t)=\pi\eta(e).$$
\begin{proof} We have
\begin{equation*}
\begin{split}
\Exp\int \im(t-z)^{-1}d\txi^{\omega,e}_k(t) &=\abs{B_k}^{-1}\Exp\im
\sum_{x\in B_k}\scalp{\delta_x}{(H_{r_k}^{\omega}-z_k)^{-1}\delta_x} \\
&=\abs{B_k}^{-1}\Exp\im
\sum_{x\in B_k}\scalp{\delta_x}{\of{(H_{r_k}^{\omega}-z_k)^{-1}-(H^{\omega}-z_k)^{-1}}\delta_x} \\
 &+\Exp\im\scalp{\delta_{x_0}}{(H^{\omega}-z_k)^{-1}\delta_{x_0}} \\
 &=I_{k,\omega}+II_{k,\omega}.
\end{split}
\end{equation*}
Now $II_{k,\omega}\To\pi\eta(e)$ by \ref{FatouPoint} and $I_{k,\omega}\To 0$, as in the proof of Proposition \ref{rescaled_measures_close}.
\end{proof}
\noindent (H3):
We need to establish that for every function $g_z(t)=\im (t-z)^{-1}$, $\im z>0$, 
\begin{equation}\label{verif_minami}
\qquad \lim_{k\To\infty}\sum_{j=1}^{n^{k-r_k}}\Exp \cI(\txi_{k,j}^{\omega,e},g_z)=0.
\end{equation}
\begin{proof} We have,
$$\abs{B_k}^{2}I(\txi_{k,j}^{\omega,e},g_z)=$$
$$ =\of{\sum_{x\in B_{k,j}} \scalp{\delta_x}{\im(H_{r_k}^{\omega}-z_k)^{-1}\delta_x}}^2-\sum_{x\in B_{k,j}} \scalp{\delta_x}{\of{\im(H_{r_k}^{\omega}-z_k)^{-1}}^2\delta_x} $$

$$=\sum_{x,y\in B_{k,j}}\det \left( \begin{array}{cc}
\scalp{\delta_x}{\im(H_{r_k}^{\omega}-z_k)^{-1}\delta_x}&\scalp{\delta_x}{\im(H_{r_k}^{\omega}-z_k)^{-1}\delta_y}  \\
\scalp{\delta_y}{\im(H_{r_k}^{\omega}-z_k)^{-1}\delta_x} & \scalp{\delta_y}{\im(H_{r_k}^{\omega}-z_k)^{-1}\delta_y} \\
\end{array} \right).$$
Using Minami's estimate \eqref{general_minami_estimate} we get the bounds
$$\abs{B_k}^{2}\Exp I(\txi_{k,j}^{\omega,e},g_z)\leq \pi^2\norm{\gamma}_{\infty}^2\abs{B_{k,j}}^2,$$
and hence
$$\sum_{j=1}^{n^{k-r_k}}\Exp I(\txi_{k,j}^{\omega,e},g_z)\leq \pi^2\norm{\gamma}_{\infty}^2 n^{-r_k},$$
which yields \eqref{verif_minami}.
\end{proof}

\appendix
\section{Minami's proof of Poisson statistics for the localized Anderson model on $\zd$}
For a rectangle $B\subset \zd$, we denote by $H_{B}^{\omega}$ the restriction of $H_\omega$ to $l^2(B)$ with Dirichlet boundary conditions: i.e. $\scalp{\delta_x}{H_B^{\omega}\delta_y}=\scalp{\delta_x}{H_{\omega}\delta_y}$ if both $x,y\in B$, and $\scalp{\delta_x}{H_B^{\omega}\delta_y}=0$ otherwise.
For $k\geq 1$, let  $B_k$ be the rectangle $\set{x\in \zd: \max_{i=1,\cdots, \mathrm{d}}\abs{x_i}\leq k  }$, and let $H_k^\omega=H_{B_k}^{\omega}$. As before, $e^{\omega,k}_1\leq e^{\omega,k}_2\leq\cdots\leq e^{\omega,k}_{\abs{B_k}},$ are the eigenvalues of $H_k^{\omega}\upharpoonright l^2(B_k)$, $\mu_k^\omega$ is the corresponding normalized counting measure given by \eqref{defn_eigenvaluecounting} and $\xi_k^{\omega,e}$ is the rescaled measure near $e$ given by \eqref{defn_rescaledmeasure}. We refer the reader to the recent work \cite{KN} for a discussion of the regime where both space and energy are rescaled. The averaged spectral measure for $H_\omega$ is given by \eqref{defn_avgspectralmeasure} and the Wegner estimate yileds that $\mu^{av}$ has a bounded density $\eta(t)$ with respect to $\Leb$. A basic result for the Anderson model is that for $\Prob$-a.e. $\omega\in\Omega$,
the spectrum of $H_\omega$ is equal to $[-2\mathrm{d},2\mathrm{d}]+\supp(\gamma)=\supp(\mu^{av})$ and $\mu_k^{\omega}$ converges to $\mu^{av}$ in the weak-* topology as $k\To\infty$ (\cite{PF,CL,CKFS}). 
\begin{thm}\label{Minami_Theorem}(Minami, 1996)
Assume that there are constants $0<C<\infty, 0<D<\infty$ and
$0<s<1$ such that
\begin{equation}\label{FM_loc}
\Exp \abs{\scalp{\delta_x}{(H_B^\omega-z)^{-1} \delta_y}}^s\leq Ce^{-D\abs{x-y}},\qquad x,y,\in\zd,
\end{equation} for
all $z$ with $e_1< \re z < e_2, \im z\neq 0$ and for all rectangles $B\subset \zd$. Assume that $e\in (e_1,e_2)$ verifies the regularity condition \eqref{FatouPoint} and that $\eta(e)>0$. Then $\xi_k^{\omega,e}$ converges to a Poisson point process on $\Real$ with intensity $\eta(e)\Leb$.
\end{thm}

Condition \eqref{FM_loc} is called fractional-moments localization. It implies that within $(e_1,e_2)$, for $\Prob$-a.e. $\omega\in\Omega$ the spectrum of $H_\omega$, if any, is pure-point with exponentially decaying eigenfunctions \cite{AM,ASFH}. For $\mathrm{d}=1$, condition \eqref{FM_loc} holds for all energy intervals $(e_1,e_2)$ \cite{Mi}. In dimensions $\mathrm{d}\geq 2$, condition $\eqref{FM_loc}$ is obtained by either moving the energy interval $(e_1,e_2)$ to $\pm\infty$ or by increasing the disorder. The two main techniques for proving that are the multiscale analysis \cite{FS,DK} and the Aizenman-Molchanov theory \cite{AM}.

\textit{Proof of Theorem \ref{Minami_Theorem}.}
We fix $\al\in(0,1)$ and for each $k$, we make a partition $$B_k=\bigcup_{j=1}^{n_k} {B_{k,j}},$$ where $B_{k,j}$ are disjoint rectangles
with side $\sim (2k)^\alpha$. Hence $n_k\sim k^{\mathrm{d}(1-\al)}$. Let $\widetilde{e}^{\omega,k,j}_{l}$, $l=1,\cdots,\abs{B_{k,j}}$ denote the eigenvalues of $H^\omega_{B_{k,j}}\upharpoonright l^2(B_{k,j})$ and let
$$\txi_{k,j}^{\omega,e}=\sum_{l=1}^{\abs{B_{k,j}}}\delta(\abs{B_k}(\widetilde{e}^{\omega,k,j}_{l}-e)),$$
$$\txi_k^{\omega,e}=\sum_{j=1}^{n_k} \txi_{k,j}^{\omega,e}.$$
Hence the point process $\txi_k^{\omega,e}$ is the sum of $n_k$ independent point processes $\txi_{k,j}^{\omega,e}$.
As in Section 3, Theorem \ref{Minami_Theorem} follows from the following two Propositions.
\begin{prop}\label{rescaled_measures_close_Minami} For every $f\in L_1(\Real,dt)$,
\begin{equation}\label{eq_same_limit_Minami}
\lim_{k\To\infty} \Exp\abs{\int f d\txi^{\omega,e}_k- \int f d\xi^{\omega,e}_k}=0.
\end{equation}
\end{prop}
\begin{prop}\label{asymptoticallyPoisson_Minami} The point process $\txi^{\omega,e}_k$ converges to a Poisson point process on $\Real$ with intensity $\eta(e)\Leb$.
\end{prop}
\newcommand{\inter}{\mathrm{int}}
\newcommand{\wall}{\mathrm{wall}}
\textit{Proof of Proposition \ref{rescaled_measures_close_Minami}.}
As in the proof of Proposition \ref{rescaled_measures_close}, it is enough to prove \eqref{eq_same_limit_Minami} for the family of functions $$g_z(t)=\im (t-z)^{-1}, \qquad \im z> 0.$$ We set 
\begin{equation}\label{def_zk_Minami}
z_k=e+\abs{B_k}^{-1}z.
\end{equation} 
Then
$$\int g_z d\txi^{\omega,e}_k- \int g_z d\xi^{\omega,e}_k$$ $$=\abs{B_k}^{-1}\im
\sum_{j=1}^{n_k}\sum_{x\in B_{k,j}}\scalp{\delta_x}{\of{(H_{B_{k,j}}^\omega-z_k)^{-1}-(H_k^\omega-z_k)^{-1}}\delta_x}.$$
Let $v_k=\beta \ln k$, where $\beta>0$ is a fixed big enough constant to be specified later. We set
$$\inter(B_{k,j})=\set{x\in B_{k,j}: dist(x,\partial B_{k,j})\geq v_k },$$ and
$$\wall(B_{k,j})=\set{x\in B_{k,j}: dist(x,\partial B_{k,j})< v_k }.$$
Then 
$$\Exp\abs{\int g_z d\txi^{\omega,e}_k- \int g_z d\xi^{\omega,e}_k}\leq \Exp\abs{I_{k,\omega}}+ \Exp\abs{II_{k,\omega}},$$
where
$$I_{k,\omega}=\abs{B_k}^{-1}\im\sum_{j=1}^{n_k}\sum_{x\in \wall(B_{k,j})}\scalp{\delta_x}{\of{(H_{B_{k,j}}^\omega-z_k)^{-1}-(H_k^\omega-z_k)^{-1}}\delta_x},$$
$$II_{k,\omega}=\abs{B_k}^{-1}\im\sum_{j=1}^{n_k}\sum_{x\in \inter(B_{k,j})}\scalp{\delta_x}{\of{(H_{B_{k,j}}^\omega-z_k)^{-1}-(H_k^\omega-z_k)^{-1}}\delta_x}.$$

The Wegner estimate \eqref{avgtxi} yields that
$$\Exp \abs{I_{k,\omega}} \leq 2\pi\norm{\gamma}_{\infty} \abs{B_k}^{-1}\sum_{j=1}^{n_k}\abs{\wall(B_{k,j})},$$
and the right hand side converges to zero as $k\To\infty$.

To estimate $\Exp\abs{II_{k,\omega}}$, we use the resolvent identity
$$\scalp{\delta_x}{\of{(H_{B_{k,j}}^\omega-z_k)^{-1}-(H_k^\omega-z_k)^{-1}}\delta_x}$$

$$=\sum_{(y,y')} \scalp{\delta_x}{(H_{B_{k,j}}^\omega-z_k)^{-1}\delta_y}\scalp{\delta_{y'}}{(H_k^\omega-z_k)^{-1}\delta_x},$$
where the sum is over all pairs $(y,y')$, with $y\in\partial B_{k,j}$, $y'\notin B_{k,j}$ and $\abs{y-y'}=1$. Hence,
\begin{equation}\label{wall_resolvent}
\Exp\abs{II_{k,\omega}}\leq \abs{B_k}^{-1}\sum_{j=1}^{n_k}\sum_{x\in \inter(B_{k,j})}\sum_{(y,y')}\Exp \abs{ \scalp{\delta_x}{(H_{B_{k,j}}^\omega-z_k)^{-1}\delta_y}\scalp{\delta_{y'}}{(H_k^\omega-z_k)^{-1}\delta_x}}.
\end{equation}
For $k$ large enough so that $e_1<\re z_k<e_2$, we use the main assumption \eqref{FM_loc} together with the bound
$$\abs{ \scalp{\delta_x}{(H_{B_{k,j}}^\omega-z_k)^{-1}\delta_y}\scalp{\delta_{y'}}{(H_k^\omega-z_k)^{-1}\delta_x}}\leq (\im z_k)^{-2}=(\abs{B_k}/\im z)^2,$$ to obtain 
\begin{equation}\label{bound_wall}
\begin{split}
&\Exp \abs{ \scalp{\delta_x}{(H_{B_{k,j}}^\omega-z_k)^{-1}\delta_y}\scalp{\delta_{y'}}{(H_k^\omega-z_k)^{-1}\delta_x}}\\
&\leq (\abs{B_k}/\im z)^{2(1-s/2)}\Exp  \abs{ \scalp{\delta_x}{(H_{B_{k,j}}^\omega-z_k)^{-1}\delta_y}\scalp{\delta_{y'}}{(H_k^\omega-z_k)^{-1}\delta_x}}^{s/2}\\
&\leq (\abs{B_k}/\im z)^{2(1-s/2)} \of{\Exp \abs{\scalp{\delta_x}{(H_{B_{k,j}}^\omega-z_k)^{-1}\delta_y}}^s}^{1/2}\of{\Exp\abs{\scalp{\delta_{y'}}{(H_k^\omega-z_k)^{-1}\delta_x}}^s}^{1/2}\\
&\leq (\abs{B_k}/\im z)^{2(1-s/2)} Ce^{-Dv_k}.\\
\end{split}
\end{equation}
Since, in \eqref{wall_resolvent}, there are $O(k^{\al(\mathrm{d}-1)})$ pairs $(y,y')$ for each $B_{k,j}$, the bounds \eqref{wall_resolvent} and \eqref{bound_wall} yield
\begin{equation*}
\begin{split}
\Exp\abs{II_{k,\omega}} &\leq O(k^{\al(\mathrm{d}-1)} \abs{B_k}^{2(1-s/2)} e^{-Dv_k})\\
                                      & = O(k^{\al(\mathrm{d}-1)+2\mathrm{d}(1-s/2)} e^{-D\beta\ln k})\\
\end{split}
\end{equation*}
Hence, if we choose $\beta>D^{-1} \of{\al(\mathrm{d}-1)+2\mathrm{d}(1-s/2)}$, then $\Exp\abs{II_{k,\omega}}\To 0$ as $k\To\infty$.
$\Box$

\textit{Proof of Proposition \ref{asymptoticallyPoisson_Minami}.}
As in the proof of Propositon \ref{asymptoticallyPoisson},
it suffices to show that $\txi^{\omega,e}_k$ and the $\txi^{\omega,e}_{k,j}$ verify the four hypotheses of Theorem \ref{GrigelionisPoisson}.
The proof of (H0), (H1) and (H3) is the same as in Propositon \ref{asymptoticallyPoisson}. It remains to show that H2 holds, i.e. for $\im z>0$,
\begin{equation}\label{H2_Minami}
\lim_{k\To\infty}\Exp\int g_z d\txi^{\omega,e}_{k}=\pi\eta(e).
\end{equation}
The argument of the proof of Proposition  \ref{rescaled_measures_close_Minami}, with $H^{\omega}_{k}$ replaced by $H_\omega$, yields that
\begin{equation}\label{H2_MinamiRep}
\lim_{k\To\infty}\Exp\of{\int g_z d\txi^{\omega,e}_{k} - \int g_z d\mu^{av}}=0,
\end{equation}
and then \eqref{H2_Minami} follows from \eqref{H2_MinamiRep} and \eqref{FatouPoint}.

$\Box$



\end{document}